\documentclass[preprint,11pt]{elsarticle}
\usepackage{amssymb}
\usepackage{booktabs}
\usepackage{url}
\usepackage{multirow}
\usepackage{hyperref}
\usepackage{amsthm}
\usepackage{graphicx}
\usepackage{subcaption}

\theoremstyle{remark}
\newtheorem{remark}{Remark}
\usepackage{doi}
\hypersetup{
colorlinks,
linkcolor={blue},
citecolor={blue},
urlcolor={red}
}

\usepackage{amsmath}
\usepackage{algorithm}
\usepackage{orcidlink}
\usepackage{color}
\usepackage{algpseudocode}
\usepackage{setspace}
\newtheorem{theorem}{Theorem}
\usepackage{lscape}

\newtheorem{corollary}{Corollary}
\usepackage[numbers]{natbib}
\setcitestyle{numbers,open={[},close={]}}
\usepackage[margin=2.8cm]{geometry}
\usepackage{lineno}
\makeatletter
\def\ps@pprintTitle{%
 \let\@oddhead\@empty
 \let\@evenhead\@empty
 \def\@oddfoot{\reset@font\hfil\thepage\hfil}
 \let\@evenfoot\@oddfoot
}
\makeatother
\begin{document}
\begin{frontmatter}

\title{Goodness-of-fit testing for the Pareto type-I distribution based on a mean residual life characterization}

\author[label1]{Shivshankar Nila\corref{cor1}\,\orcidlink{0009-0004-0465-5490}}
\ead{shivshankarnila1@gmail.com}
\author[label1]{Ishapathik Das} 
\cortext[cor1]{Corresponding author} 
\author[label1]{N. Balakrishna}

\address[label1]{Department of Mathematics and Statistics, Indian Institute of Technology Tirupati, Tirupati, India.} 

\doublespacing
\begin{abstract}
\noindent The statistical analysis of heavy-tailed data has received considerable attention because extreme observations frequently arise in many practical applications. The Pareto type-I distribution is a fundamental heavy-tailed model used in economics, finance, actuarial science, insurance, reliability, and extreme value analysis. In this paper, we propose novel goodness-of-fit tests for the Pareto distribution using a mean residual life characterization. The test statistic is constructed using U-statistic theory, and its asymptotic behaviour is established under both the null and alternative hypotheses. Its finite-sample performance is evaluated through Monte Carlo simulations using maximum-likelihood and method-of-moments estimation and compared with existing tests. The results show that the proposed test controls the nominal significance level and performs competitively in terms of power across a broad range of alternatives. Finally, the proposed methodology is illustrated using the Danish fire insurance loss and pollution datasets.
\end{abstract}

\end{frontmatter}

\doublespacing
\section{Introduction}\label{sec1}
\noindent Over the last decade, there has been increasing interest in modeling extreme events exhibiting heavy-tailed behaviour using extreme-value distributions \cite{coles2001introduction,nila2025modeling}, particularly for quantifying risk measures such as Value-at-Risk (VaR) and Expected Shortfall (ES). Among heavy-tailed models, the Pareto distribution is one of the most widely used and also plays an important role in extreme value theory; see \cite{beirlant2006statistics}. It was first introduced by \cite{pareto1897new} and has found extensive applications in economics \cite{piketty2014capital}, insurance claim modeling \cite{brazauskas2000robust,gay2004pricing}, finance, actuarial science, and reliability; see \cite{ismail2004simple}. The Pareto distribution has several generalizations, including the Pareto types II, III, and IV distributions, as well as the generalized Pareto distribution (GPD). A detailed discussion of these distributions and their interrelationships can be found in \cite{arnold}. It is also a well-established result in extreme value theory that, under suitable regularity conditions, the distribution of exceedances over a sufficiently high threshold converges to the GPD; see \cite{coles2001introduction,nila2026flexible,Pickands1975}.

Assessing whether observed data follow a specified distribution is an important problem in statistical analysis, and goodness-of-fit (GoF) tests provide a natural framework for this purpose. Classical GoF procedures include tests based on the empirical distribution function (EDF), such as the Kolmogorov--Smirnov (KS), Cramér--von Mises (CvM), and Anderson--Darling (AD) tests. An alternative approach is based on distributional characterizations; a characterization of a distribution is a property that uniquely identifies that distribution among all probability distributions \cite{nikitin2017tests}. Characterization-based GoF procedures use such identifying properties to construct test statistics for assessing departures from a specified assumed distribution. This approach has been widely studied in the literature; see \cite{galambos2006characterizations,nikitin2017tests}.

Several GoF tests have been proposed for the Pareto type I distribution. In particular, characterization-based procedures have been developed using different distributional properties; see \cite{allison2022distribution,ndwandwe2023testing,ngatchou2024classes}. Related tests based on different characterizations of the Pareto distribution have also been studied by \cite{akbari2020characterization}, \cite{bojana2016},\cite{obradovic2015goodness}, and \cite{volkova2016goodness}. 
In particular, different test constructions may offer varying power, robustness, or computational efficiency across scenarios.  Comprehensive reviews and comparisons of GoF tests for the Pareto distribution can be found in \cite{chu2019review,ndwandwe2023testing}. Recently, GoF tests for the Pareto type-I distribution based on Stein-type characterizations have been proposed by \cite{avhad2026goodness} and \cite{bhati2025new}. 
The Lomax distribution, also known as the Pareto type-II distribution, is widely used for modeling heavy-tailed data. Recently, a GoF test for this distribution was proposed by \cite{abhijithkrishnu2026goodness}.

In this paper, we propose a new GoF test for the Pareto type I distribution based on a characterization involving the mean residual life (MRL) function. It is well known that the MRL function characterizes the underlying distribution under suitable conditions; see \cite{arnold,laurent1974characterization}. 
\citet{cox1962renewal} (page 128) has the result that the conditional expectation \(E(X\mid X>x)\) for a positive random variable $X$ (when it exists)
characterizes the distribution of the random variable.
 Motivated by this result, we use the conditional expectation \(E(X\mid X>x)\) as the characterization quantity for constructing our goodness-of-fit test.
The MRL function for $X$ is defined as
\[
m(x)=E(X-x\mid X>x)
=\frac{1}{\bar{F}(x)}\int_x^\infty \bar{F}(t)\,dt,\qquad x>0,
\]
where $\bar{F}(x)=1-F(x)$ denotes the survival function. For the Pareto type I distribution with finite mean, the MRL function has a simple linear form, which provides a natural basis for constructing a GoF procedure.
The proposed test statistic is formulated using the theory of $U$-statistics \cite{lee2019u}, and its asymptotic distribution is derived under the null and alternative hypotheses. The proposed statistic admits a simple and computationally efficient form, facilitating its practical implementation.

The remainder of this paper is organized as follows. Section~\ref{sec2} outlines the MRL-based characterizations and describes the formulation of the test statistic along with its asymptotic properties under the null and alternative hypotheses. Section~\ref{SStudy} presents a comprehensive Monte Carlo simulation study to evaluate the empirical size and power of the proposed test under different parameter settings and estimation techniques and compares it with existing methods. Section~\ref{sec} illustrates the practical application of the proposed test using two real-world datasets. Finally, Section~\ref{conclusion} provides some concluding remarks and future directions, and the proofs of the main results are presented in the Appendix~\ref{proof:appendix}.

\section{Overview of characterization and test statistics}\label{sec2}
\noindent In this section, we introduce a new GoF test for the Pareto Type I distribution based on MRL characterizations, which can be used to characterize distributions under suitable conditions. The distribution function (DF) $F$ and probability density function (PDF) $f$ of the Pareto type-I distribution are given by
\begin{equation}
F(x)=1-\left(\frac{\sigma}{x}\right)^{\alpha}, \qquad
x\geq \sigma,\; \alpha>0,
\end{equation}
and
\begin{equation}
\label{pdf:Pareto}
f(x)=\frac{\alpha \sigma^{\alpha}}{x^{\alpha+1}}, \qquad
x\geq \sigma,\; \alpha>0,~~\text{respectively .}
\end{equation}
The MRL function for the Pareto type I distribution is defined by
\begin{equation}
m(x)=\frac{x}{\alpha-1},
\qquad x\geq \sigma,\quad \alpha>1,
\end{equation}
where $\sigma>0$ is the scale parameter. Throughout this paper, we assume that $\sigma=1$ and that the shape parameter $\alpha>1$, so that the distribution has a finite mean.  A detailed discussion can be found in \cite{ndwandwe2023testing}.
We denote by $\mathcal{P}:=\{P(\alpha):\alpha>1\}$ the family of Pareto type I distributions with shape parameter $\alpha>1$ and support $[1,\infty)$. We will now outline the characterization theorem employed in the formulation of our test. 
\begin{theorem}\label{thm0}
[\citet{10.1214/aos/1176342723}]
Let $X$ be a positive random variable whose mean exists, and has a Pareto distribution if, and only if,
\[E(X\mid X>t)=\mu+\beta t,\qquad \beta>1,\]
where $\mu$ and $\beta$ are constants.
\end{theorem}
\begin{proof}
    For the proof, refer to \citet{10.1214/aos/1176342723}.
\end{proof}
\noindent \citet{10.1214/aos/1176342723} noted that the Pareto law may be subdivided into two types depending on the location parameter, namely, Pareto type-I and Pareto type-II. 
In this paper, we restrict attention to the Pareto type-I distribution under the assumptions stated above.
For simplicity, the Pareto type I distribution will hereafter be referred to as the Pareto distribution. Therefore, we have the following characterization.
\begin{theorem}\label{thm01}
Let \(X\) be a positive random variable with \(E(X)<\infty\).
Then \(X\) follows a Pareto distribution with parameter
\(\alpha>1\) if and only if
\[E(X\mid X>t)=\frac{\alpha}{\alpha-1}t,\qquad t>1.\]
\end{theorem}
\begin{proof}
    For the proof, refer to \citet{10.1214/aos/1176342723}.
\end{proof}
We now obtain the following characterization.
\begin{theorem}\label{thm1}
Let $X$ be a continuous random variable with support $[1,\infty)$ and
$E(X)<\infty$. Then $X$ has a Pareto distribution with shape
parameter $\alpha>1$ if and only if
\[
E\left[\left(X-t-\frac{t}{\alpha-1}\right)I(X>t)\right]=0,
\qquad t> 1.
\]
\end{theorem}
\begin{proof}
The proof directly follows from Theorems~\ref{thm0} and~\ref{thm01}. An alternative proof is provided in Appendix~\ref{proof:thm1}.
\end{proof}
The characterization established in Theorem~\ref{thm1} can be used to construct a GoF test, as it uniquely identifies the underlying distribution through an expectation identity. Based on this characterization, we develop the GoF test.
Let $X,X_1,X_2,\ldots,X_n$ be a random sample from an unknown distribution $F$; we are interested in the following composite hypothesis testing problem:
\[H_0:F\in\mathcal{P}\quad \text{against} \quad H_1:F\notin\mathcal{P}.
\]
 For testing the above hypothesis, we define a departure measure that discriminates between the null and alternative hypotheses. The departure measure $\Delta(F)$ is given by
\begin{equation}
\Delta(F)=\int_{1}^{\infty}
E\left[\left(X-t-\frac{t}{\alpha-1}\right)I(X>t)\right]\,dF(t).
\end{equation}

According to Theorem~\ref{thm1}, $\Delta(F)=0$ under the null hypothesis, and under a fixed alternative for which $\Delta(F)\neq 0$, the quantity $\Delta(F)$ measures the departure from the null hypothesis. To construct a test statistic using the theory of $U$-statistics, we express $\Delta(F)$ in terms of expectations of functions of random variables. Using Theorem~\ref{thm1}, we obtain
\begin{align}
    \Delta(F) &=\int_1^\infty \left(E\left[\left(X-t-\frac{t}{\alpha-1}\right)I(X>t)\right]\right)dF(t) \nonumber\\
    &=\int_1^\infty \int_1^\infty (x-t)I(x>t)\,dF(x)dF(t)
    -\int_1^\infty \int_1^\infty \frac{t}{\alpha-1}I(x>t)\,dF(x)dF(t) \nonumber\\
    &=\int_1^\infty \int_t^\infty (x-t)\,dF(x)dF(t)
    -\int_1^\infty \int_t^\infty \frac{t}{\alpha-1}\,dF(x)dF(t) \nonumber\\
    &= E\{(X_2-X_1) I(X_2\geq X_1)\}
    - \frac{1}{\alpha-1}E\{X_1 I(X_2\geq X_1)\}.
\end{align}
Next, we obtain the test statistic using the theory of $U$-statistics. For this, we consider a symmetric kernel, $h_1(X_1,X_2)$, of degree 2, as
$h_1(X_1,X_2)=\frac{|X_1-X_2|}{2}$.
Then a $U$-statistic defined by
\[U_1=\frac{2}{n(n-1)}\sum_{i=1}^{n}
\sum_{j=1,j<i}^{n}
h_1(X_i,X_j),\]
is an unbiased estimator of
\[E\left(\frac{|X_1-X_2|}{2}\right)=\theta_1.\]
\noindent Similarly, consider the symmetric kernel~$
h_2(X_1,X_2)=\frac{\min(X_1,X_2)}{2}$, 
and the corresponding U-statistic is
\[U_2=
\frac{2}{n(n-1)}
\sum_{i=1}^{n}
\sum_{j=1,j<i}^{n}
h_2(X_i,X_j),\]
which is an unbiased estimator of
\[E\left(\frac{\min(X_1,X_2)}{2}\right)=\theta_2.\]
Here $\alpha$ is an unknown parameter and, for the Pareto distribution, we consider the maximum likelihood estimator (MLE) and the method of moments estimator (MME) of the shape parameter $\alpha$, and are given as
\[\widehat{\alpha}_{\mathrm{MLE}}=
\frac{n}{\sum_{i=1}^{n}\log X_i},
\qquad
\widehat{\alpha}_{\mathrm{MME}}
=\frac{\bar X}{\bar X-1},\qquad~\text{where}~~
\bar X=\frac{1}{n}\sum_{i=1}^{n}X_i.\]
Various consistent estimators for $\alpha$, have been proposed in the literature; details are discussed in \cite{quandt1964old}.
Replacing $\alpha$ with its consistent estimator $\widehat{\alpha}$, the test statistic is given by
\begin{equation}
\label{Test_Statistics}
\widehat{\Delta}=U_1-\frac{1}{\widehat{\alpha}-1}U_2,
\end{equation}
For sufficiently large values of $|\widehat{\Delta}|$, we reject the null hypothesis $H_0$ against the alternative hypothesis $H_1$.
\begin{remark}
Under standard regularity conditions, $\widehat{\alpha}$ can be considered the U-statistic, provided that the underlying distribution possesses finite moments of the required order. The consistency and asymptotic normality of the U-statistic-based quantities rely on these moment assumptions. Consequently, the applicability of the proposed procedure is restricted to distributions satisfying appropriate moment conditions.
\end{remark}
Next, we study the asymptotic properties of the test statistic. Since $U_1$ and $U_2$ are $U$-statistics, they are consistent estimators of
$E\left(\frac{|X_1-X_2|}{2}\right)$ and
$E\left(\frac{\min(X_1,X_2)}{2}\right)$, respectively, see \cite{lehmann1951consistency}. 
When $\widehat{\alpha}$ is the MLE of $\alpha$, by the strong law of large numbers and standard properties of the MLE,~ $\widehat{\alpha}\xrightarrow{\mathrm{a.s.}}\alpha$ as $n\to\infty$~ (\cite{avhad2026goodness}).
Hence, the following result is straightforward.

\begin{theorem}\label{thm2}
Let $\widehat{\alpha}$ ~be a consistent estimator of $\alpha$,~ under $H_1$, as $n\to\infty$,
$\widehat{\Delta}$ converges in probability to $\Delta(F)$, where
\[\Delta(F)=\theta_1-\frac{\theta_2}{\alpha-1}.\]
\end{theorem}
\begin{proof}
The proof can be found in Appendix~\ref{proof:thm2}.
\end{proof}
\noindent
To derive the asymptotic distribution of the proposed test statistic, let $\widehat{\alpha}$ be a consistent estimator of $\alpha$ admitting the asymptotic linear representation
\[\sqrt{n}\left(\widehat{\alpha}-\alpha\right)=
\frac{1}{\sqrt{n}}\sum_{j=1}^{n}\ell(X_j)+o_P(1),\]
where $\ell$ denotes the corresponding influence function satisfying $E[\ell(X_1)]=0$ and $E[\ell^2(X_1)]<\infty$. In particular, for the MLE, the associated influence function is
$\ell(x)=\alpha\left(1-\alpha\log x\right)$.
Thus, the MLE satisfies
\[\sqrt{n}\left(\widehat{\alpha}_{\mathrm{MLE}}-\alpha\right)
=\frac{1}{\sqrt{n}}\sum_{j=1}^{n}\alpha\left(1-\alpha\log X_j\right)
+o_P(1).\]
For more, see \cite{avhad2026goodness,henze2024asymptotic,lee2019u}.

\begin{theorem}\label{thm3}
Under the standing assumptions, as $n\to\infty$,
$\sqrt{n}\left(\widehat{\Delta}-\Delta(F)\right)
\xrightarrow{d} N(0,\sigma^2)$,
where
\begin{equation}
\label{u_var}
\sigma^2=
\operatorname{Var}\!\left[
2h_1^{(1)}(X_1)
-\frac{2}{\alpha-1}h_2^{(1)}(X_1)
+\frac{\theta_2}{(\alpha-1)^2}\ell(X_1)
\right],
\end{equation}
with
\[h_i^{(1)}(X_1)=E\!\left[h_i(X_1,X_2)\mid X_1\right]-\theta_i,
\qquad i=1,2,\]
and
$\theta_1=E[h_1(X_1,X_2)]$ and
$\theta_2=E[h_2(X_1,X_2)]$.
\end{theorem}

\begin{proof}
The proof can be found in Appendix~\ref{proof:thm3}.
\end{proof}

\begin{corollary}
\label{corl1}
Under $H_0$, as $n\to\infty$, 
$\sqrt{n}\,\widehat{\Delta}
\xrightarrow{d}
N(0,\sigma_0^2)$,
where $\sigma_0^2$ denotes the asymptotic variance obtained by evaluating \eqref{u_var} under $H_0$.
\end{corollary}

\begin{proof}
The proof can be found in Appendix~\ref{proof:corl1}.
\end{proof}

The asymptotic critical region for the proposed test can be obtained using Corollary~\ref{corl1}. We reject the null hypothesis
$H_0$ against the alternative hypothesis $H_1$ at a significance level $\delta$, if
\[\frac{\sqrt{n}\,|\widehat{\Delta}|}{\widehat{\sigma}_0}>
Z_{\delta/2},\]
where $Z_{\delta}$ is the upper $\delta$-percentile point of a standard normal distribution.
Since it is difficult to obtain a consistent estimator $\widehat{\sigma}_0^{\,2}$ of $\sigma_0^2$ and the corresponding critical values for $\widehat{\Delta}$, we follow the simulated critical region (SCR) approach of \cite{xavier2025goodness} to determine the critical values of the test. It minimises the dependence on the asymptotic critical value. Specifically, lower and upper quantiles are established based on the exact distribution, ensuring lower $(c_1)$ and upper $(c_2)$ quantiles in such a way that
\[P(\widehat{\Delta}<c_1)=P(\widehat{\Delta}>c_2)=\delta/2.\]
The finite sample performance of the test is evaluated through a Monte Carlo simulation
study and the results are reported in Section \ref{SStudy}.

\section{Simulation study}\label{SStudy}
\noindent In this section, we evaluate the finite-sample performance of the proposed test procedure by conducting an extensive Monte Carlo simulation study using \texttt{R} software. To demonstrate the competitiveness and effectiveness of our newly developed method compared with existing methods for the Pareto distribution, the following GoF tests are considered for comparison.
 \begin{itemize} 
\item Based on a Stein-type  characterization of the Pareto distribution, \cite{avhad2026goodness} proposed an integral-type GoF test. The test statistic is given by
\begin{equation*}
\widehat{G}_{I}= \left((\widehat{\alpha}+1)U+\frac{\widehat{\alpha}}{2}
\right),
\end{equation*}
\begin{equation*}
    U = \dfrac{2}{n(n-1)} \sum_{i=1  }^{n} \sum_{j < i, j=1}^{n} \frac{1}{2}
\left[\frac{X_j}{X_i}I{X_j<X_i}+
\frac{X_i}{X_j}I{X_i<X_j}
-\frac{1}{X_i}-\frac{1}{X_j}
\right]
\end{equation*}
The null hypothesis is rejected for large values of
$|\widehat{\Delta}_{I}|$.
\item \cite{avhad2026goodness} also proposed a Cramér-von Mises type test based on the same characterization. Let the notation  $
a_{ij}=\min(X_i,X_j)-1,
\qquad
b_{ij}=\max(X_i,X_j)$.

\[U_1=\frac{1}{\binom{n}{3}}
\sum_{1\le i<j<k\le n}
\left(\frac{a_{ik}a_{jk}}{X_iX_j}
+
\frac{a_{ij}a_{jk}}{X_iX_k}
+
\frac{a_{ij}a_{ik}}{X_jX_k}
\right),\qquad U_3
=\frac1n\sum_{i=1}^{n}X_i^{-1}, ~\text{and}\]
\[U_2=
\frac{1}{\binom{n}{3}}
\sum_{1\le i<j<k\le n}
\left(
\frac{b_{ij}}{X_k}I(b_{ij}\le X_k)
+
\frac{b_{ik}}{X_j}I(b_{ik}\le X_j)
+
\frac{b_{jk}}{X_i}I(b_{jk}\le X_i)
\right).
\]

Then the test statistic is

\[\widehat{G}_{M}=
(\widehat{\alpha}+1)^2U_1
-(\widehat{\alpha}+1)(U_2-U_3)
-\frac{2\widehat{\alpha}+1}{3}.\]
Here, $\widehat{\alpha}$ denotes a consistent estimator of the Pareto shape parameter $\alpha$, obtained either by the MLE or the MME. The null hypothesis is rejected for large values of
$\widehat{\Delta}_{M}$.

\item Let \(X,X_1,\ldots,X_n\) be independent and identically distributed (i.i.d.) positive continuous random variables with common distribution function \(F\). Then, for every integer \(k\) such that \(2\leq k\leq n\),
~~$X^{1/k}\stackrel{d}{=}\min(X_1,\ldots,X_k)$
holds if and only if \(F\) is the Pareto distribution.
Based on this characterization, \cite{allison2022distribution} proposed three GoF tests for the Pareto distribution, with test statistics given by
             \begin{align*}
                  I_{n,k}&= \int_1^\infty \Delta_{n,k}(x)dF_n(x),\\
             \end{align*}
             where the discrepancy corresponding to \(X^{1/k}\) is defined as 
             \begin{small}
             \begin{equation*}
                 \Delta_{n,k}(x)= \dfrac{1}{n}\sum_{j=1}^n I \Big\{X_{j}^{k^{-1}} \leq x\Big\}- \dfrac{1}{n^k} \sum_{j_1,\ldots, j_k=1}^n I\{\min(X_{j_1}, \ldots, X_{j_k}) \leq x \}. 
             \end{equation*}
              \end{small}
          For the Monte Carlo study, we set the tuning parameter to \(k=2\).

\item Let \(X,X_1,\ldots,X_n\) be i.i.d. non-negative random variables with distribution function \(F\). For any integer \(k\) satisfying \(2\leq k\leq n\), the random variables \(X^{1/k}\) and \(\min\{X_1,\ldots,X_k\}\) have the same distribution if and only if, for every \(t\in\mathcal{R}\),
        \begin{equation*}
            E\Big\{ \dfrac{1}{k} \exp \Big(-it X^{k^{-1}}\Big) -[1-F(X)]^{k-1} \exp(-itX)  \Big\}= 0,
        \end{equation*}
where \(F\) denotes the Pareto distribution function given in (\ref{pdf:Pareto}). Based on this characterization, \cite{ngatchou2024classes} proposed the test statistic
\begin{equation*}
    T_{k,n,w}= \int_{\mathcal{R}} |S_{k,n,\widehat{\alpha}_n}(t) |^2 w(t) dt,
\end{equation*}
here for all $t\in\mathcal{R}$,
\begin{equation*}
    S_{k,n,\alpha}(t)= \dfrac{1}{\sqrt{n}}\sum_{j=1}^n \Big[ \dfrac{1}{k}\exp(-itX_j^{k^{-1}})-X_j^{-\alpha(k-1)}\exp(-itX_j)\Big].
\end{equation*} 
     In the practical application, we choose $a=0.5$ for the weight function. The two test statistics are defined as
\begin{enumerate}
    \item \(T_{k,a}^{(1)}\) corresponds to \(w(t)=e^{-a|t|}\) (Laplace weight).
\item \(T_{k,a}^{(2)}\) corresponds to \(w(t)=e^{-at^2}\) (normal weight).
\end{enumerate}

  \item Based on likelihood ratio, \cite{zhang2002powerful} proposed two tests with test statistics given by
        \begin{equation*}
            ZA= -\sum_{j=1}^n \Bigg[\dfrac{\log \big(F(X_{(j)})\big)}{n-j+0.5}+ \dfrac{\log \big(1-F(X_{(j)})\big)}{j-0.5} \Bigg]
        \end{equation*}
and 
\begin{equation*}
            ZB= \sum_{j=1}^n \Bigg[ \log \Bigg( \dfrac{ \big(F(X_{(j)}) \big)^{-1} -1 }{  (n-0.5)/(j-0.75)-1}   \Bigg) \Bigg]^2,
        \end{equation*}
where \(X_{(i)}\) denotes the \(i\)-th order statistic of the random sample \(X_1,\ldots,X_n\) from \(F\).
\item  Kolmogorov-Smirnov (KS) test statistic 
\begin{equation*}
KS=\sup_{x\geq1}\big|F_n(x)-F(x)\big|,
\qquad
~~\text{where}~~F_n(x)=\frac{1}{n}\sum_{i=1}^{n} I(X_i\leq x).
\end{equation*}
  \item Anderson-Darling (AD) test statistic
    \begin{equation*}
        AD= \int \dfrac{\big(F_n(x) - {F}(x)\big)^2}{F(x)(1-F(x))}dF(x). 
    \end{equation*}

In terms of the order statistics, the above test statistic can be expressed 
     \begin{equation*}
         AD= -n-\dfrac{1}{n}\sum_{i=1}^n (2i-1) \big[\log \big(F(X_{(i)})\big)+ \log (1- F(X_{(n+1-i)})) \big].
     \end{equation*}

    \item Cramér--von Mises (CvM) test statistic,
 
 \hspace{2.7cm}    $CvM$= $\int \big(F_n(x)-{F}(x) \big)^2 dF(x)$.

In terms of the order statistics, the above test statistic can be expressed 
     
    \begin{equation*}
        CvM= \sum_{i=1}^n \Big[ F(X_{(i)})- \dfrac{2i-1}{2n} \Big]^2 +\dfrac{1}{12n}, 
    \end{equation*}
where $X_{(j)}$ denotes the order statistics.

\end{itemize}

For each test, the empirical critical value is estimated from \(10{,}000\) Monte Carlo replications at the \(0.05\) significance level.
The shape parameter of the Pareto distribution is estimated using either the MLE or MME, given by \(\widehat{\alpha}=n/\sum_{i=1}^{n}\log X_i\) and \(\widehat{\alpha}=\bar{X}/(\bar{X}-1)\), respectively.
To estimate the empirical type I error and powers of the considered tests, sample sizes of $n = 25, 30, 50, 75$, and $100$ are used. A wide range of alternative distributions is considered and presented in Table~\ref{table:1} for power comparison purposes. Samples are generated from the Pareto distribution in Section~\ref{pdf:Pareto} for different values of the shape parameter to evaluate the empirical size of the proposed test. The empirical size and power of the proposed test $\widehat{\Delta}$, together with several existing GoF tests, at the 5\% significance level are presented in Tables~\ref{table:size}--\ref{table:power6}. All calculations and simulations were performed in \texttt{R} software \cite{team2020ra}.

The Monte Carlo simulation results (Tables~\ref{table:size}--\ref{table:power6}) show that $\widehat{\Delta}$ preserves the nominal type-I error rate under the null Pareto distribution for all sample sizes and parameter values considered. Under the alternative distributions listed in Table~\ref{table:1}, including gamma, log-logistic, Lévy, tilted Pareto, exponential, Benini, log-Weibull, and Dhillon, the proposed test exhibits good empirical power that generally increases with sample size. The highest empirical power among the tests considered is highlighted in bold in each row of Tables ~\ref{table:power3}--\ref{table:power6}. 
Overall, from reported simulation results, $\widehat{\Delta}$ performs competitively with the existing GoF tests across a broad range of alternatives, although some competing tests exhibit comparable or superior power for certain alternatives and parameter-estimation settings. We observe that empirical power increases with sample size, as expected. Table~\ref{tab:com1} summaries the proposed and existing GoF tests in terms of their underlying principles, statistic types, asymptotic null distributions, and computational features. The proposed test is based on an MRL characterization, has a normal asymptotic null distribution, and is simple and fast to compute.
These findings confirm the effectiveness and robustness of the proposed test across most of the alternatives considered.

\begin{small}
\begin{table}[!t]
\centering
\caption{Summary of various choices of alternative distributions.}
\vspace{0.1cm}
\label{table:1} 
\fontsize{9pt}{11pt}\selectfont
\begin{tabular}{p{4cm} p{7cm} p{2cm}} 
\hline 
Distribution & Form of density function & Notation \\[1ex]
\hline 
Gamma &
$\displaystyle \frac{1}{\Gamma(\theta)}
(x-1)^{\theta-1}e^{-(x-1)}$
& $\Gamma(\theta)$ \\[2ex]

Log-logistic &
$\beta
\left({x-1}\right)^{\beta-1}
\left[1+\left({x-1}\right)^{\beta}\right]^{-2}$
& $LL(\beta)$ \\[2ex]

L\'evy &
$\sqrt{\dfrac{c}{2\pi}}
\dfrac{\exp\!\left[-\dfrac{c}{2(x-1)}\right]}
{(x-1)^{3/2}}$

& $L(c)$ \\[2ex]

Tilted Pareto &
$\dfrac{\lambda+1}{(x+\lambda)^2}$
& $TP(\lambda)$ \\[2ex]

Exponential &
$\lambda e^{-\lambda (x-1)}$
& $E(\lambda)$ \\[2ex]

Benini &
$\dfrac{e^{-\theta (\log x)^2}}{x^2}(1+2\theta \log x)$
& $BN(\theta)$ \\[2ex]

Log-Weibull &
$\dfrac{(1+\theta)(\log x)^{\theta}}{x^{2+\theta}}$
& $LW(\theta)$ \\[2ex]

Dhillon &
$\dfrac{\theta+1}{x}
\exp\!\left\{-(\log x)^{\theta+1}\right\}
(\log x)^{\theta},\quad$
& $DH(\theta)$ \\[2ex]
\hline
\end{tabular}
\end{table}
\end{small}

\begin{table}[!t]
\centering
\caption{Empirical type I error rates at $0.05$ level of significance}
\label{table:size}
\resizebox{16cm}{!}{
\begin{tabular}{c c c c c c c c c c c  c c  }
\hline
\multicolumn{13}{c}{\textbf{Maximum Likelihood Estimator (MLE)}}\\
\hline
Null & $n$ & $\widehat{\Delta}$ & KS & AD & CvM & $G_I$ &$G_M$ &   $T_{2,0.5}^{(1)}$ &  $T_{2,0.5}^{(2)}$ &$ I_{n,2}$  & $ZA$  & $ZB$ \\
\hline
\multirow{5}{*}{$P(1.5)$}
& 25 & 0.047 & 0.053 & 0.052 & 0.050 & 0.052 & 0.035 & 0.049 & 0.051 & 0.054 & 0.053 & 0.052 \\
& 30 & 0.052 & 0.058 & 0.059 & 0.064 & 0.058 & 0.061 & 0.052 & 0.053 & 0.044 & 0.054 & 0.059 \\
& 50 & 0.055 & 0.051 & 0.050 & 0.051 & 0.055 & 0.034 & 0.052 & 0.048 & 0.053 & 0.052 & 0.049 \\
& 75 & 0.053 & 0.053 & 0.048 & 0.048 & 0.064 & 0.046 & 0.051 & 0.049 & 0.056 & 0.065 & 0.060 \\
& 100 & 0.049 & 0.052 & 0.051 & 0.052 & 0.047 & 0.049 & 0.052 & 0.058 & 0.049 & 0.051 & 0.058 \\
\hline
\multirow{5}{*}{$P(2.0)$}
& 25 & 0.055 & 0.055 & 0.048 & 0.052 & 0.053 & 0.053 & 0.046 & 0.053 & 0.047 & 0.042 & 0.046 \\
& 30 & 0.052 & 0.050 & 0.049 & 0.055 & 0.052 & 0.054 & 0.053 & 0.053 & 0.047 & 0.050 & 0.050 \\
& 50 & 0.046 & 0.049 & 0.044 & 0.047 & 0.055 & 0.062 & 0.046 & 0.040 & 0.051 & 0.037 & 0.042 \\
& 75 & 0.048 & 0.060 & 0.061 & 0.059 & 0.070 & 0.051 & 0.059 & 0.051 & 0.062 & 0.047 & 0.045 \\
& 100 & 0.052 & 0.043 & 0.042 & 0.043 & 0.052 & 0.064 & 0.049 & 0.045 & 0.055 & 0.051 & 0.044 \\
\hline

\multirow{5}{*}{$P(3)$}
& 25 & 0.058 & 0.050 & 0.060 & 0.058 & 0.054 & 0.051 & 0.052 & 0.047 & 0.042 & 0.049 & 0.050 \\
& 30 & 0.054 & 0.050 & 0.049 & 0.053 & 0.040 & 0.049 & 0.044 & 0.042 & 0.055 & 0.048 & 0.051 \\
& 50 & 0.043 & 0.057 & 0.051 & 0.052 & 0.048 & 0.069 & 0.051 & 0.058 & 0.049 & 0.050 & 0.055 \\
& 75 & 0.054 & 0.062 & 0.059 & 0.054 & 0.052 & 0.053 & 0.051 & 0.048 & 0.054 & 0.051 & 0.049 \\
& 100 & 0.050 & 0.047 & 0.043 & 0.046 & 0.049 & 0.053 & 0.046 & 0.049 & 0.044 & 0.046 & 0.048 \\
\hline
\multicolumn{13}{c}{\textbf{Method of Moments Estimator (MME)}}\\
\hline
\multirow{5}{*}{$P(1.5)$}
& 25 & 0.059 & 0.061 & 0.048 & 0.047 & 0.045 & 0.039 & 0.051 & 0.051 & 0.044 & 0.055 & 0.053 \\
& 30 & 0.046 & 0.063 & 0.063 & 0.058 & 0.048 & 0.034 & 0.045 & 0.039 & 0.046 & 0.051 & 0.048 \\
& 50 & 0.048 & 0.043 & 0.043 & 0.040 & 0.040 & 0.033 & 0.045 & 0.049 & 0.055 & 0.065 & 0.058 \\
& 75 & 0.049 & 0.052 & 0.045 & 0.048 & 0.048 & 0.062 & 0.051 & 0.054 & 0.053 & 0.057 & 0.055 \\
& 100 & 0.053 & 0.053 & 0.047 & 0.049 & 0.057 & 0.047 & 0.044 & 0.040 & 0.051 & 0.053 & 0.053 \\
\hline

\multirow{5}{*}{$P(2.0)$}
& 25 & 0.056 & 0.047 & 0.051 & 0.053 & 0.052 & 0.063 & 0.052 & 0.053 & 0.042 & 0.041 & 0.044 \\
& 30 & 0.045 & 0.056 & 0.058 & 0.054 & 0.038 & 0.062 & 0.053 & 0.046 & 0.045 & 0.058 & 0.061 \\
& 50 & 0.049 & 0.065 & 0.056 & 0.058 & 0.054 & 0.055 & 0.047 & 0.054 & 0.056 & 0.049 & 0.049 \\
& 75 & 0.054 & 0.049 & 0.049 & 0.049 & 0.054 & 0.051 & 0.059 & 0.053 & 0.051 & 0.054 & 0.058 \\
& 100 & 0.055 & 0.048 & 0.044 & 0.040 & 0.050 & 0.066 & 0.049 & 0.053 & 0.053 & 0.043 & 0.046 \\
\hline

\multirow{5}{*}{$P(3)$}
& 25 & 0.052 & 0.060 & 0.061 & 0.060 & 0.063 & 0.055 & 0.049 & 0.046 & 0.049 & 0.046 & 0.043 \\
& 30 & 0.050 & 0.041 & 0.047 & 0.043 & 0.046 & 0.050 & 0.046 & 0.049 & 0.054 & 0.059 & 0.063 \\
& 50 & 0.059 & 0.047 & 0.045 & 0.048 & 0.039 & 0.050 & 0.054 & 0.054 & 0.050 & 0.050 & 0.050 \\
& 75 & 0.044 & 0.048 & 0.051 & 0.050 & 0.051 & 0.062 & 0.058 & 0.066 & 0.044 & 0.056 & 0.053 \\
& 100 & 0.052 & 0.047 & 0.044 & 0.047 & 0.050 & 0.050 & 0.054 & 0.055 & 0.057 & 0.055 & 0.056 \\
\hline
\end{tabular}}
\end{table}

\begin{table}[!t]
\centering
\caption{Empirical power at $0.05$ level of significance based on MLE}
\label{table:power3}
\resizebox{16cm}{!}{
\begin{tabular}{c c c c c c c c c c c  c c  }
\hline
Alternative & $n$ & $\widehat{\Delta}$ & KS & AD & CvM & $G_I$ &$G_M$ &   $T_{2,0.5}^{(1)}$ &  $T_{2,0.5}^{(2)}$ &$ I_{n,2}$ & $ZA$  & $ZB$ \\
\hline
\multirow{5}{*}{$\Gamma(1)$}
& 25 & 0.309 & 0.283 & 0.326 & 0.351 & 0.444 & 0.431 & 0.405 & \textbf{0.515} & 0.338 & 0.348 & 0.354 \\
& 30 & 0.435 & 0.376 & 0.376 & 0.442 & 0.541 & 0.486 & 0.469 & \textbf{0.551} & 0.350 & 0.414 & 0.419 \\
& 50 & 0.666 & 0.565 & 0.611 & 0.651 & 0.740 & 0.726 & 0.690 & \textbf{0.823} & 0.584 & 0.668 & 0.628 \\
& 75 & 0.818 & 0.767 & 0.829 & 0.851 & 0.913 & 0.890 & 0.858 & \textbf{0.938} & 0.761 & 0.863 & 0.836 \\
& 100 & 0.927 & 0.869 & 0.948 & 0.952 & 0.962 & 0.979 & 0.947 & \textbf{0.990} & 0.894 & 0.968 & 0.964 \\
\hline
\multirow{5}{*}{$\Gamma(1.2)$}
& 25 & 0.766 & 0.534 & 0.608 & 0.666 & 0.768 & 0.754 & 0.740 & \textbf{0.847} & 0.621 & 0.655 & 0.662 \\
& 30 & 0.863 & 0.620 & 0.740 & 0.759 & 0.854 & 0.840 & 0.769 & \textbf{0.885} & 0.712 & 0.783 & 0.772 \\
& 50 & 0.965 & 0.891 & 0.942 & 0.951 & 0.964 & 0.961 & 0.927 & \textbf{0.984} & 0.903 & 0.955 & 0.948 \\
& 75 & \textbf{1.000} & 0.964 & 0.993 & 0.993 & 0.998 & 0.999 & 0.992 & 0.999 & 0.985 & 0.998 & 0.997 \\
& 100 & \textbf{1.000} & 0.996 & \textbf{1.000} & \textbf{1.000} & \textbf{1.000} & \textbf{1.000} & \textbf{1.000} & \textbf{1.000} & 0.999 & \textbf{1.000} & \textbf{1.000} \\
\hline
\multirow{5}{*}{$TP(0.5)$}
& 25 & \textbf{0.994} & 0.073 & 0.076 & 0.082 & 0.107 & 0.136 & 0.344 & 0.384 & 0.086 & 0.110 & 0.101 \\
& 30 & \textbf{0.997} & 0.083 & 0.078 & 0.092 & 0.115 & 0.137 & 0.321 & 0.415 & 0.085 & 0.095 & 0.102 \\
& 50 & \textbf{0.999} & 0.115 & 0.114 & 0.139 & 0.150 & 0.146 & 0.389 & 0.469 & 0.108 & 0.122 & 0.105 \\
& 75 & \textbf{1.000} & 0.154 & 0.151 & 0.163 & 0.170 & 0.182 & 0.412 & 0.489 & 0.161 & 0.158 & 0.142 \\
& 100 & \textbf{1.000} & 0.158 & 0.162 & 0.176 & 0.211 & 0.242 & 0.427 & 0.532 & 0.192 & 0.188 & 0.192 \\
\hline

\multirow{5}{*}{$TP(0.8)$}
& 25 & \textbf{1.000} & 0.085 & 0.089 & 0.103 & 0.175 & 0.143 & 0.404 & 0.515 & 0.108 & 0.149 & 0.137 \\
& 30 & \textbf{0.999} & 0.117 & 0.098 & 0.128 & 0.127 & 0.157 & 0.398 & 0.527 & 0.113 & 0.142 & 0.124 \\
& 50 & \textbf{1.000} & 0.172 & 0.163 & 0.188 & 0.212 & 0.303 & 0.496 & 0.598 & 0.190 & 0.170 & 0.175 \\
& 75 & \textbf{1.000} & 0.211 & 0.233 & 0.266 & 0.284 & 0.329 & 0.496 & 0.617 & 0.252 & 0.272 & 0.251 \\
& 100 & \textbf{1.000} & 0.286 & 0.297 & 0.317 & 0.358 & 0.375 & 0.592 & 0.730 & 0.347 & 0.335 & 0.302 \\
\hline

\multirow{5}{*}{$LL(1.2)$}
& 25 & \textbf{0.915} & 0.107 & 0.102 & 0.120 & 0.158 & 0.185 & 0.297 & 0.337 & 0.126 & 0.169 & 0.140 \\
& 30 & \textbf{0.940} & 0.119 & 0.121 & 0.130 & 0.170 & 0.198 & 0.299 & 0.349 & 0.152 & 0.213 & 0.203 \\
& 50 & \textbf{0.978} & 0.182 & 0.218 & 0.227 & 0.301 & 0.272 & 0.385 & 0.456 & 0.235 & 0.279 & 0.263 \\
& 75 & \textbf{0.987} & 0.255 & 0.331 & 0.321 & 0.378 & 0.405 & 0.471 & 0.544 & 0.325 & 0.354 & 0.336 \\
& 100 & \textbf{0.993} & 0.289 & 0.412 & 0.388 & 0.459 & 0.465 & 0.527 & 0.599 & 0.447 & 0.444 & 0.427 \\
\hline

\multirow{5}{*}{$LL(1.5)$}
& 25 & \textbf{0.938} & 0.436 & 0.428 & 0.464 & 0.498 & 0.608 & 0.612 & 0.662 & 0.483 & 0.564 & 0.539 \\
& 30 & \textbf{0.969} & 0.498 & 0.547 & 0.587 & 0.624 & 0.662 & 0.657 & 0.723 & 0.560 & 0.676 & 0.650 \\
& 50 & \textbf{0.989} & 0.697 & 0.766 & 0.754 & 0.860 & 0.852 & 0.800 & 0.867 & 0.809 & 0.818 & 0.803 \\
& 75 & \textbf{0.995} & 0.893 & 0.955 & 0.942 & 0.961 & 0.963 & 0.935 & 0.965 & 0.947 & 0.959 & 0.955 \\
& 100 & \textbf{0.997} & 0.955 & 0.987 & 0.981 & 0.979 & 0.985 & 0.989 & 0.991 & 0.991 & 0.989 & 0.988 \\
\hline
\multirow{5}{*}{$DH(0.5)$}
& 25 & \textbf{0.980} & 0.483 & 0.533 & 0.574 & 0.697 & 0.664 & 0.678 & 0.777 & 0.576 & 0.624 & 0.591 \\
& 30 & \textbf{0.994} & 0.530 & 0.645 & 0.643 & 0.796 & 0.760 & 0.729 & 0.864 & 0.693 & 0.711 & 0.700 \\
& 50 & \textbf{0.999} & 0.824 & 0.915 & 0.908 & 0.937 & 0.944 & 0.912 & 0.963 & 0.900 & 0.924 & 0.924 \\
& 75 & \textbf{1.000} & 0.924 & 0.979 & 0.976 & 0.992 & 0.991 & 0.976 & 0.992 & 0.985 & 0.984 & 0.987 \\
& 100 & \textbf{1.000} & 0.980 & 0.996 & 0.995 & 0.999 & 0.999 & 0.996 & \textbf{1.000} & 0.999 & 0.998 & 0.998 \\
\hline
\multirow{5}{*}{$DH(0.8)$}
& 25 & \textbf{0.998} & 0.805 & 0.886 & 0.888 & 0.959 & 0.940 & 0.919 & 0.963 & 0.892 & 0.926 & 0.925 \\
& 30 & \textbf{1.000} & 0.898 & 0.962 & 0.964 & 0.977 & 0.968 & 0.955 & 0.986 & 0.954 & 0.966 & 0.970 \\
& 50 & \textbf{1.000} & 0.984 & 0.997 & 0.997 & \textbf{1.000} & \textbf{1.000} & 0.998 & \textbf{1.000} & 0.999 & 0.997 & 0.998 \\
& 75 & \textbf{1.000} & 0.999 & \textbf{1.000} & \textbf{1.000} & \textbf{1.000} & \textbf{1.000} & \textbf{1.000} & \textbf{1.000} & \textbf{1.000} & \textbf{1.000} & \textbf{1.000} \\
& 100 & \textbf{1.000} & \textbf{1.000} & \textbf{1.000} & \textbf{1.000} & \textbf{1.000} & \textbf{1.000} & \textbf{1.000} & \textbf{1.000} & \textbf{1.000} & \textbf{1.000} & \textbf{1.000} \\
\hline
\end{tabular}}
\end{table}

\begin{table}[!t]
\centering
\caption{Empirical power at $0.05$ level of significance based on MLE}
\label{table:power4}
\resizebox{16cm}{!}{
\begin{tabular}{c c c c c c c c c c c c c c }
\hline
Alternative & $n$ & $\widehat{\Delta}$ & KS & AD & CvM & $G_I$ &$G_M$ &   $T_{2,0.5}^{(1)}$ &  $T_{2,0.5}^{(2)}$ &$ I_{n,2}$ & $ZA$  & $ZB$ \\
\hline

\multirow{5}{*}{$L(2)$}
& 25 & \textbf{1.000} & 0.252 & 0.300 & 0.273 & 0.307 & 0.477 & 0.878 & 0.937 & 0.336 & 0.686 & 0.561 \\
& 30 & \textbf{1.000} & 0.314 & 0.397 & 0.363 & 0.395 & 0.561 & 0.877 & 0.930 & 0.385 & 0.816 & 0.665 \\
& 50 & \textbf{1.000} & 0.526 & 0.682 & 0.528 & 0.632 & 0.855 & 0.958 & 0.980 & 0.668 & 0.986 & 0.954 \\
& 75 & \textbf{1.000} & 0.832 & 0.942 & 0.821 & 0.885 & 0.960 & 0.999 & 0.999 & 0.899 & \textbf{1.000} & 0.999 \\
& 100 & \textbf{1.000} & 0.955 & 0.988 & 0.934 & 0.971 & 0.996 & \textbf{1.000} & \textbf{1.000} & 0.964 & \textbf{1.000} & \textbf{1.000} \\
\hline
\multirow{5}{*}{$L(3)$}
& 25 & \textbf{1.000} & 0.509 & 0.597 & 0.571 & 0.598 & 0.762 & 0.955 & 0.984 & 0.597 & 0.897 & 0.797 \\
& 30 & \textbf{1.000} & 0.606 & 0.701 & 0.648 & 0.698 & 0.871 & 0.965 & 0.994 & 0.692 & 0.963 & 0.891 \\
& 50 & \textbf{1.000} & 0.898 & 0.964 & 0.921 & 0.949 & 0.987 & 0.999 & \textbf{1.000} & 0.946 & \textbf{1.000} & \textbf{1.000} \\
& 75 & \textbf{1.000} & 0.993 & 0.999 & 0.988 & 0.997 & \textbf{1.000} & \textbf{1.000} & \textbf{1.000} & 0.996 & \textbf{1.000} & \textbf{1.000} \\
& 100 & \textbf{1.000} & \textbf{1.000} & \textbf{1.000} & \textbf{1.000} & \textbf{1.000} & \textbf{1.000} & \textbf{1.000} & \textbf{1.000} & \textbf{1.000} & \textbf{1.000} & \textbf{1.000} \\
\hline
\multirow{5}{*}{$E(0.5)$}
& 25 & \textbf{0.991} & 0.518 & 0.611 & 0.651 & 0.739 & 0.718 & 0.769 & 0.893 & 0.610 & 0.646 & 0.605 \\
& 30 & \textbf{0.997} & 0.617 & 0.737 & 0.791 & 0.799 & 0.791 & 0.810 & 0.930 & 0.710 & 0.744 & 0.735 \\
& 50 & \textbf{1.000} & 0.890 & 0.940 & 0.948 & 0.955 & 0.953 & 0.954 & 0.995 & 0.899 & 0.938 & 0.916 \\
& 75 & \textbf{1.000} & 0.967 & 0.995 & 0.997 & 0.997 & 0.998 & 0.995 & \textbf{1.000} & 0.980 & 0.997 & 0.994 \\
& 100 & \textbf{1.000} & 0.994 & \textbf{1.000} & \textbf{1.000} & \textbf{1.000} & 0.999 & 0.999 & \textbf{1.000} & 0.996 & \textbf{1.000} & \textbf{1.000} \\
\hline
\multirow{5}{*}{$E(0.8)$}
& 25 & \textbf{0.701} & 0.404 & 0.423 & 0.494 & 0.566 & 0.500 & 0.537 & 0.672 & 0.392 & 0.453 & 0.442 \\
& 30 & \textbf{0.784} & 0.382 & 0.454 & 0.488 & 0.630 & 0.568 & 0.577 & 0.745 & 0.459 & 0.544 & 0.531 \\
& 50 & \textbf{0.950} & 0.638 & 0.737 & 0.769 & 0.816 & 0.839 & 0.802 & 0.921 & 0.703 & 0.800 & 0.767 \\
& 75 & \textbf{0.990} & 0.844 & 0.932 & 0.940 & 0.956 & 0.946 & 0.943 & 0.986 & 0.878 & 0.960 & 0.936 \\
& 100 & \textbf{0.999} & 0.945 & 0.985 & 0.987 & 0.993 & 0.984 & 0.986 & \textbf{0.999} & 0.946 & 0.993 & 0.990 \\
\hline
\multirow{5}{*}{$BN(0.3)$}
& 25 & \textbf{0.692} & 0.151 & 0.138 & 0.170 & 0.179 & 0.189 & 0.255 & 0.322 & 0.142 & 0.147 & 0.133 \\
& 30 & \textbf{0.725} & 0.128 & 0.117 & 0.140 & 0.240 & 0.224 & 0.236 & 0.333 & 0.145 & 0.162 & 0.150 \\
& 50 & \textbf{0.888} & 0.212 & 0.226 & 0.253 & 0.275 & 0.282 & 0.334 & 0.462 & 0.252 & 0.269 & 0.250 \\
& 75 & \textbf{0.956} & 0.322 & 0.381 & 0.408 & 0.458 & 0.452 & 0.435 & 0.602 & 0.315 & 0.374 & 0.338 \\
& 100 & \textbf{0.982} & 0.429 & 0.515 & 0.519 & 0.558 & 0.528 & 0.499 & 0.655 & 0.411 & 0.525 & 0.506 \\
\hline

\multirow{5}{*}{$BN(0.7)$}
& 25 & 0.275 & 0.250 & 0.252 & 0.270 & 0.350 & 0.346 & 0.338 & \textbf{0.387} & 0.274 & 0.253 & 0.259 \\
& 30 & 0.421 & 0.271 & 0.295 & 0.345 & 0.341 & 0.364 & 0.336 & \textbf{0.436} & 0.278 & 0.323 & 0.316 \\
& 50 & 0.577 & 0.432 & 0.466 & 0.528 & 0.563 & 0.576 & 0.514 & \textbf{0.630} & 0.446 & 0.523 & 0.501 \\
& 75 & 0.741 & 0.562 & 0.642 & 0.682 & 0.801 & 0.769 & 0.705 & \textbf{0.821} & 0.655 & 0.625 & 0.588 \\
& 100 & 0.886 & 0.747 & 0.831 & 0.850 & 0.883 & 0.872 & 0.792 & \textbf{0.902} & 0.737 & 0.833 & 0.812 \\
\hline

\multirow{5}{*}{$LW(0.5)$}
& 25 & \textbf{0.979} & 0.186 & 0.210 & 0.236 & 0.326 & 0.319 & 0.395 & 0.471 & 0.249 & 0.323 & 0.297 \\
& 30 & \textbf{0.989} & 0.228 & 0.233 & 0.262 & 0.339 & 0.350 & 0.482 & 0.552 & 0.301 & 0.384 & 0.342 \\
& 50 & \textbf{0.997} & 0.353 & 0.427 & 0.417 & 0.560 & 0.536 & 0.567 & 0.668 & 0.484 & 0.543 & 0.504 \\
& 75 & \textbf{1.000} & 0.541 & 0.644 & 0.621 & 0.657 & 0.730 & 0.721 & 0.802 & 0.642 & 0.710 & 0.703 \\
& 100 & \textbf{1.000} & 0.656 & 0.771 & 0.743 & 0.812 & 0.843 & 0.792 & 0.884 & 0.781 & 0.842 & 0.839 \\
\hline
\multirow{5}{*}{$LW(0.7)$}
& 25 & \textbf{0.990} & 0.344 & 0.327 & 0.376 & 0.516 & 0.479 & 0.570 & 0.650 & 0.374 & 0.459 & 0.440 \\
& 30 & \textbf{0.993} & 0.355 & 0.417 & 0.426 & 0.533 & 0.576 & 0.602 & 0.709 & 0.483 & 0.588 & 0.528 \\
& 50 & \textbf{1.000} & 0.578 & 0.680 & 0.654 & 0.789 & 0.775 & 0.754 & 0.839 & 0.706 & 0.789 & 0.769 \\
& 75 & \textbf{1.000} & 0.792 & 0.907 & 0.873 & 0.924 & 0.914 & 0.895 & 0.946 & 0.894 & 0.924 & 0.914 \\
& 100 & \textbf{1.000} & 0.897 & 0.969 & 0.947 & 0.973 & 0.977 & 0.950 & 0.973 & 0.962 & 0.978 & 0.980 \\
\hline
\end{tabular}}
\end{table}

\begin{table}[!t]
\centering
\caption{Empirical power at $0.05$ level of significance based on MME}
\label{table:power5}
\resizebox{16cm}{!}{
\begin{tabular}{c c c c c c c c c c c  c c  }
\hline
Alternative & $n$ & $\widehat{\Delta}$ & KS & AD & CvM & $G_I$ &$G_M$ &   $T_{2,0.5}^{(1)}$ &  $T_{2,0.5}^{(2)}$ &$ I_{n,2}$ & $ZA$  & $ZB$ \\
\hline
\multirow{5}{*}{$\Gamma(1)$}
& 25 & 0.448 & 0.461 & 0.518 & 0.518 & \textbf{0.564} & 0.563 & 0.465 & 0.549 & 0.318 & 0.404 & 0.432 \\
& 30 & 0.473 & 0.553 & 0.602 & 0.609 & \textbf{0.645} & 0.577 & 0.544 & 0.599 & 0.353 & 0.486 & 0.548 \\
& 50 & 0.695 & 0.728 & 0.822 & 0.819 & 0.826 & 0.849 & 0.791 & \textbf{0.872} & 0.607 & 0.706 & 0.769 \\
& 75 & 0.907 & 0.923 & 0.953 & 0.957 & 0.960 & \textbf{0.966} & 0.927 & 0.950 & 0.759 & 0.875 & 0.892 \\
& 100 & 0.981 & 0.959 & 0.986 & 0.987 & 0.989 & 0.987 & 0.987 & \textbf{0.993} & 0.899 & 0.948 & 0.962 \\
\hline
\multirow{5}{*}{$\Gamma(1.2)$}
& 25 & 0.776 & 0.765 & 0.837 & 0.839 & 0.858 & 0.830 & 0.789 & \textbf{0.860} & 0.626 & 0.760 & 0.791 \\
& 30 & 0.886 & 0.856 & 0.902 & 0.906 & 0.915 & 0.904 & 0.862 & \textbf{0.920} & 0.727 & 0.834 & 0.882 \\
& 50 & 0.991 & 0.979 & 0.991 & 0.992 & \textbf{0.994} & 0.986 & 0.981 & 0.992 & 0.916 & 0.967 & 0.978 \\
& 75 & \textbf{1.000} & 0.996 & 0.999 & 0.999 & \textbf{1.000} & \textbf{1.000} & \textbf{1.000} & \textbf{1.000} & 0.989 & 0.998 & \textbf{1.000} \\
& 100 & \textbf{1.000} & \textbf{1.000} & \textbf{1.000} & \textbf{1.000} & \textbf{1.000} & \textbf{1.000} & \textbf{1.000} & \textbf{1.000} & 0.999 & 0.999 & \textbf{1.000} \\
\hline
\multirow{5}{*}{$TP(0.5)$}
& 25 & \textbf{0.965} & 0.518 & 0.629 & 0.551 & 0.309 & 0.312 & 0.299 & 0.391 & 0.076 & 0.507 & 0.560 \\
& 30 & \textbf{0.980} & 0.526 & 0.684 & 0.623 & 0.356 & 0.373 & 0.339 & 0.440 & 0.074 & 0.565 & 0.598 \\
& 50 & \textbf{0.995} & 0.659 & 0.783 & 0.727 & 0.453 & 0.472 & 0.437 & 0.554 & 0.116 & 0.730 & 0.784 \\
& 75 & \textbf{1.000} & 0.759 & 0.874 & 0.831 & 0.578 & 0.631 & 0.493 & 0.637 & 0.142 & 0.818 & 0.866 \\
& 100 & \textbf{1.000} & 0.841 & 0.920 & 0.886 & 0.684 & 0.740 & 0.657 & 0.776 & 0.192 & 0.900 & 0.923 \\
\hline

\multirow{5}{*}{$TP(0.8)$}
& 25 & \textbf{0.994} & 0.634 & 0.781 & 0.722 & 0.438 & 0.484 & 0.360 & 0.530 & 0.112 & 0.697 & 0.721 \\
& 30 & \textbf{0.996} & 0.692 & 0.843 & 0.794 & 0.498 & 0.546 & 0.439 & 0.594 & 0.117 & 0.760 & 0.774 \\
& 50 & \textbf{1.000} & 0.840 & 0.924 & 0.896 & 0.662 & 0.670 & 0.589 & 0.718 & 0.184 & 0.878 & 0.907 \\
& 75 & \textbf{1.000} & 0.938 & 0.976 & 0.970 & 0.808 & 0.848 & 0.766 & 0.864 & 0.256 & 0.975 & 0.979 \\
& 100 & \textbf{1.000} & 0.979 & 0.995 & 0.988 & 0.890 & 0.910 & 0.818 & 0.923 & 0.330 & 0.991 & 0.992 \\
\hline

\multirow{5}{*}{$LL( 1.2)$}
& 25 & \textbf{0.797} & 0.298 & 0.366 & 0.361 & 0.262 & 0.304 & 0.304 & 0.350 & 0.140 & 0.366 & 0.342 \\
& 30 & \textbf{0.822} & 0.331 & 0.451 & 0.429 & 0.367 & 0.317 & 0.318 & 0.372 & 0.165 & 0.396 & 0.413 \\
& 50 & \textbf{0.926} & 0.407 & 0.504 & 0.465 & 0.423 & 0.428 & 0.388 & 0.468 & 0.209 & 0.507 & 0.511 \\
& 75 & \textbf{0.966} & 0.532 & 0.638 & 0.619 & 0.521 & 0.541 & 0.476 & 0.536 & 0.344 & 0.654 & 0.642 \\
& 100 & \textbf{0.974} & 0.621 & 0.705 & 0.684 & 0.624 & 0.586 & 0.576 & 0.611 & 0.412 & 0.701 & 0.703 \\
\hline

\multirow{5}{*}{$LL( 1.5)$}
& 25 & \textbf{0.898} & 0.618 & 0.718 & 0.712 & 0.692 & 0.699 & 0.648 & 0.694 & 0.483 & 0.724 & 0.702 \\
& 30 & \textbf{0.919} & 0.744 & 0.821 & 0.808 & 0.752 & 0.743 & 0.677 & 0.740 & 0.566 & 0.814 & 0.791 \\
& 50 & \textbf{0.972} & 0.882 & 0.922 & 0.903 & 0.905 & 0.892 & 0.884 & 0.890 & 0.825 & 0.949 & 0.939 \\
& 75 & \textbf{0.989} & 0.952 & 0.972 & 0.962 & 0.962 & 0.967 & 0.966 & 0.969 & 0.957 & 0.987 & 0.982 \\
& 100 & 0.990 & 0.982 & 0.992 & 0.988 & 0.993 & \textbf{0.997} & 0.992 & 0.988 & 0.995 & \textbf{0.997} & \textbf{0.997} \\
\hline
\multirow{5}{*}{$DH(0.5)$}
& 25 & \textbf{0.978} & 0.789 & 0.871 & 0.867 & 0.871 & 0.842 & 0.745 & 0.838 & 0.550 & 0.818 & 0.826 \\
& 30 & \textbf{0.981} & 0.855 & 0.921 & 0.913 & 0.901 & 0.902 & 0.823 & 0.911 & 0.653 & 0.897 & 0.899 \\
& 50 & \textbf{1.000} & 0.975 & 0.991 & 0.990 & 0.989 & 0.980 & 0.963 & 0.984 & 0.895 & 0.985 & 0.985 \\
& 75 & \textbf{1.000} & 0.997 & \textbf{1.000} & \textbf{1.000} & 0.998 & 0.999 & 0.996 & \textbf{1.000} & 0.982 & 0.998 & 0.998 \\
& 100 & \textbf{1.000} & \textbf{1.000} & \textbf{1.000} & \textbf{1.000} & \textbf{1.000} & \textbf{1.000} & \textbf{1.000} & \textbf{1.000} & 0.998 & \textbf{1.000} & \textbf{1.000} \\
\hline
\multirow{5}{*}{$DH(0.8)$}
& 25 & \textbf{0.999} & 0.959 & 0.989 & 0.988 & 0.987 & 0.984 & 0.952 & 0.982 & 0.887 & 0.982 & 0.983 \\
& 30 & \textbf{0.998} & 0.982 & 0.995 & 0.995 & 0.995 & 0.995 & 0.981 & 0.996 & 0.948 & 0.995 & 0.996 \\
& 50 & \textbf{1.000} & \textbf{1.000} & \textbf{1.000} & \textbf{1.000} & \textbf{1.000} & \textbf{1.000} & \textbf{1.000} & \textbf{1.000} & 0.998 & \textbf{1.000} & \textbf{1.000} \\
& 75 & \textbf{1.000} & \textbf{1.000} & \textbf{1.000} & \textbf{1.000} & \textbf{1.000} & \textbf{1.000} & \textbf{1.000} & \textbf{1.000} & \textbf{1.000} & \textbf{1.000} & \textbf{1.000} \\
& 100 & \textbf{1.000} & \textbf{1.000} & \textbf{1.000} & \textbf{1.000} & \textbf{1.000} & \textbf{1.000} & \textbf{1.000} & \textbf{1.000} & \textbf{1.000} & \textbf{1.000} & \textbf{1.000} \\
\hline
\end{tabular}}
\end{table}

\begin{table}[!t]
\centering
\caption{Empirical power at $0.05$ level of significance based on MME}
\label{table:power6}
\resizebox{16cm}{!}{
\begin{tabular}{c c c c c c c c c c c c c c }
\hline
Alternative & $n$ & $\widehat{\Delta}$ & KS & AD & CvM & $G_I$ &$G_M$ &   $T_{2,0.5}^{(1)}$ &  $T_{2,0.5}^{(2)}$ &$ I_{n,2}$ & $ZA$  & $ZB$ \\
\hline
\multirow{5}{*}{$L( 2)$}
& 25 & \textbf{1.000} & 0.997 & \textbf{1.000} & \textbf{1.000} & 0.991 & 0.995 & 0.921 & 0.976 & 0.322 & \textbf{1.000} & \textbf{1.000} \\
& 30 & \textbf{1.000} & \textbf{1.000} & \textbf{1.000} & \textbf{1.000} & 0.997 & 0.998 & 0.967 & 0.991 & 0.402 & \textbf{1.000} & \textbf{1.000} \\
& 50 & \textbf{1.000} & \textbf{1.000} & \textbf{1.000} & \textbf{1.000} & \textbf{1.000} & \textbf{1.000} & 0.999 & \textbf{1.000} & 0.654 & \textbf{1.000} & \textbf{1.000} \\
& 75 & \textbf{1.000} & \textbf{1.000} & \textbf{1.000} & \textbf{1.000} & \textbf{1.000} & \textbf{1.000} & \textbf{1.000} & \textbf{1.000} & 0.890 & \textbf{1.000} & \textbf{1.000} \\
& 100 & \textbf{1.000} & \textbf{1.000} & \textbf{1.000} & \textbf{1.000} & \textbf{1.000} & \textbf{1.000} & \textbf{1.000} & \textbf{1.000} & 0.974 & \textbf{1.000} & \textbf{1.000} \\
\hline
\multirow{5}{*}{$L(3)$}
& 25 & \textbf{1.000} & \textbf{1.000} & \textbf{1.000} & \textbf{1.000} & \textbf{1.000} & \textbf{1.000} & 0.992 & 0.999 & 0.605 & \textbf{1.000} & \textbf{1.000} \\
& 30 & \textbf{1.000} & \textbf{1.000} & \textbf{1.000} & \textbf{1.000} & \textbf{1.000} & \textbf{1.000} & 0.997 & \textbf{1.000} & 0.685 & \textbf{1.000} & \textbf{1.000} \\
& 50 & \textbf{1.000} & \textbf{1.000} & \textbf{1.000} & \textbf{1.000} & \textbf{1.000} & \textbf{1.000} & \textbf{1.000} & \textbf{1.000} & 0.940 & \textbf{1.000} & \textbf{1.000} \\
& 75 & \textbf{1.000} & \textbf{1.000} & \textbf{1.000} & \textbf{1.000} & \textbf{1.000} & \textbf{1.000} & \textbf{1.000} & \textbf{1.000} & 0.995 & \textbf{1.000} & \textbf{1.000} \\
& 100 & \textbf{1.000} & \textbf{1.000} & \textbf{1.000} & \textbf{1.000} & \textbf{1.000} & \textbf{1.000} & \textbf{1.000} & \textbf{1.000} & \textbf{1.000} & \textbf{1.000} & \textbf{1.000} \\
\hline
\multirow{5}{*}{$E(0.5)$}
& 25 & \textbf{0.986} & 0.879 & 0.940 & 0.935 & 0.894 & 0.888 & 0.868 & 0.938 & 0.593 & 0.814 & 0.852 \\
& 30 & \textbf{0.998} & 0.921 & 0.956 & 0.948 & 0.942 & 0.923 & 0.907 & 0.968 & 0.676 & 0.884 & 0.900 \\
& 50 & \textbf{1.000} & 0.989 & \textbf{1.000} & 0.999 & 0.995 & 0.993 & 0.993 & 0.999 & 0.894 & 0.993 & 0.997 \\
& 75 & \textbf{1.000} & \textbf{1.000} & \textbf{1.000} & \textbf{1.000} & \textbf{1.000} & \textbf{1.000} & \textbf{1.000} & \textbf{1.000} & 0.983 & \textbf{1.000} & \textbf{1.000} \\
& 100 & \textbf{1.000} & \textbf{1.000} & \textbf{1.000} & \textbf{1.000} & \textbf{1.000} & \textbf{1.000} & \textbf{1.000} & \textbf{1.000} & 0.993 & \textbf{1.000} & \textbf{1.000} \\
\hline

\multirow{5}{*}{$E(0.8)$}
& 25 & \textbf{0.739} & 0.566 & 0.634 & 0.654 & 0.689 & 0.679 & 0.620 & 0.729 & 0.376 & 0.480 & 0.531 \\
& 30 & \textbf{0.814} & 0.692 & 0.758 & 0.751 & 0.787 & 0.757 & 0.653 & 0.765 & 0.437 & 0.591 & 0.644 \\
& 50 & \textbf{0.952} & 0.882 & 0.931 & 0.933 & 0.924 & 0.940 & 0.910 & \textbf{0.952} & 0.701 & 0.848 & 0.886 \\
& 75 & \textbf{0.997} & 0.971 & 0.988 & 0.988 & 0.987 & 0.993 & 0.980 & 0.992 & 0.875 & 0.953 & 0.973 \\
& 100 & 0.999 & 0.996 & \textbf{1.000} & \textbf{1.000} & 0.999 & 0.998 & 0.997 & \textbf{1.000} & 0.953 & 0.994 & 0.996 \\
\hline

\multirow{5}{*}{$BN(0.3)$}
& 25 & \textbf{0.567} & 0.290 & 0.375 & 0.364 & 0.289 & 0.292 & 0.227 & 0.333 & 0.116 & 0.252 & 0.254 \\
& 30 & \textbf{0.623} & 0.326 & 0.405 & 0.414 & 0.360 & 0.348 & 0.299 & 0.378 & 0.145 & 0.252 & 0.271 \\
& 50 & \textbf{0.821} & 0.469 & 0.553 & 0.544 & 0.507 & 0.502 & 0.488 & 0.535 & 0.234 & 0.381 & 0.413 \\
& 75 & \textbf{0.927} & 0.590 & 0.678 & 0.694 & 0.602 & 0.620 & 0.608 & 0.713 & 0.345 & 0.486 & 0.524 \\
& 100 & \textbf{0.965} & 0.728 & 0.801 & 0.802 & 0.740 & 0.767 & 0.696 & 0.780 & 0.379 & 0.679 & 0.739 \\
\hline

\multirow{5}{*}{$BN(0.7)$}
& 25 & 0.348 & 0.387 & 0.403 & 0.429 & \textbf{0.434} & 0.432 & 0.366 & 0.430 & 0.236 & 0.333 & 0.365 \\
& 30 & 0.398 & 0.418 & 0.440 & 0.478 & \textbf{0.539} & 0.492 & 0.401 & 0.495 & 0.309 & 0.369 & 0.387 \\
& 50 & 0.628 & 0.597 & 0.690 & 0.706 & \textbf{0.711} & \textbf{0.711} & 0.561 & 0.654 & 0.437 & 0.545 & 0.566 \\
& 75 & 0.798 & 0.789 & 0.835 & 0.850 & 0.861 & \textbf{0.873} & 0.806 & 0.858 & 0.629 & 0.720 & 0.773 \\
& 100 & 0.928 & 0.877 & 0.931 & 0.936 & 0.919 & 0.934 & 0.903 & \textbf{0.949} & 0.780 & 0.852 & 0.878 \\
\hline

\multirow{5}{*}{$LW(0.5)$}
& 25 & \textbf{0.932} & 0.520 & 0.679 & 0.655 & 0.505 & 0.538 & 0.499 & 0.562 & 0.236 & 0.646 & 0.640 \\
& 30 & \textbf{0.955} & 0.536 & 0.688 & 0.653 & 0.610 & 0.582 & 0.530 & 0.622 & 0.307 & 0.681 & 0.675 \\
& 50 & \textbf{0.989} & 0.708 & 0.823 & 0.804 & 0.805 & 0.770 & 0.671 & 0.736 & 0.453 & 0.881 & 0.874 \\
& 75 & \textbf{0.996} & 0.847 & 0.929 & 0.917 & 0.886 & 0.880 & 0.788 & 0.854 & 0.634 & 0.933 & 0.939 \\
& 100 & \textbf{0.997} & 0.929 & 0.964 & 0.956 & 0.964 & 0.953 & 0.913 & 0.943 & 0.800 & 0.978 & 0.975 \\
\hline
\multirow{5}{*}{$LW(0.7)$}
& 25 & \textbf{0.965} & 0.654 & 0.769 & 0.750 & 0.687 & 0.701 & 0.614 & 0.699 & 0.391 & 0.784 & 0.768 \\
& 30 & \textbf{0.984} & 0.727 & 0.843 & 0.818 & 0.781 & 0.756 & 0.670 & 0.767 & 0.440 & 0.873 & 0.853 \\
& 50 & \textbf{0.995} & 0.872 & 0.935 & 0.921 & 0.935 & 0.902 & 0.843 & 0.891 & 0.710 & 0.954 & 0.955 \\
& 75 & \textbf{0.999} & 0.966 & 0.990 & 0.986 & 0.977 & 0.974 & 0.965 & 0.981 & 0.907 & 0.990 & 0.991 \\
& 100 & \textbf{1.000} & 0.988 & 0.998 & 0.995 & 0.996 & 0.981 & 0.988 & 0.991 & 0.973 & 0.998 & 0.998 \\
\hline
\end{tabular}}
\end{table}

\begin{table}[!t]
\centering
\caption{Summary of the proposed and existing goodness-of-fit tests for the Pareto distribution.}
\label{tab:com1}
\resizebox{17cm}{!}{
\begin{tabular}{p{2cm} p{3cm} p{3.1cm}p{2.4cm} p{6.5cm}}
\hline
\textbf{Test} & \textbf{Underlying principle} & \textbf{Test statistic type} & \textbf{Asymptotic null dist.}  & \textbf{Computational features} \\[1ex]
\hline
$\widehat{\Delta}$ & Mean residual life characterization & $U$-statistic (asymptotic test) & Normal &    Easy to compute; fast. \\
\hline
$\widehat{G}_I$ & Stein's identity characterization & Integral-type statistic & Normal   &  Easy to compute; fast. \\
\hline
$\widehat{G}_M$  & Stein's identity characterization & Cram\'er--von Mises type & Normal   & Slow to compute (triple summation). \\
\hline
$T_{n,a}^{(1)}$
& Weighted characteristic function (CF) approach
& Integrated squared deviation of empirical CF
& Weighted $\chi^2_1$
& Simple to compute; test statistic is computed using closed-form kernel expressions; tuning parameter $a$ controls sensitivity \\[1ex]
\hline
$T_{n,a}^{(2)}$
& Weighted CF approach
& Quadratic functional of empirical CF
& Weighted $\chi^2_1$
& Computationally efficient; significantly faster than bootstrap-based procedures for moderate and large samples; good small-sample behaviour \\[1ex]
\hline
$ I_{n,2}$ & Order statistics characterization & Integral-type statistic & Normal &    Fast to compute. \\
\hline
$ZA$ & Likelihood-ratio test &Empirical Distribution Function (EDF)-based statistic & Non-standard  & Very fast and simple. \\ 
\hline
$ZB$ & Likelihood-ratio test & EDF-based statistic & Non-standard &             Very fast and simple. \\[1ex]
\hline
\end{tabular}}
\label{all_test_table}
\end{table}

\section{Data analysis}\label{sec}
\noindent In this section, we illustrate the application of the proposed test using two real-world datasets.
\subsection{Illustration 1}
\noindent 
We consider the well-known Danish fire insurance claims dataset, one of the most widely studied benchmark datasets in the literature, especially in actuarial. The dataset consists of large fire insurance claims in Denmark from $3$ January $1980$ to $31$ December $1990$ and is
available from the GitHub repository \href{https://github.com/Shivshankarnila/MRL-GOF-Test/tree/main}{Danish Fire Insurance Claims Dataset}. Standard insurance datasets share many common characteristics with the Danish fire losses; consequently, this dataset has become a popular benchmark for evaluating distributional fitting and GoF procedures.
In this analysis, we focus on claim amounts exceeding the threshold of $13.50$ million DKK and standardize these observations by dividing each by $13.50$.
The resulting exceedance sample, having a sample size of $72$, is then used to assess the GoF of the Pareto distribution.

To assess the adequacy of the Pareto model, we have also presented the MRL (mean excess) plot and the empirical distribution function with the fitted Pareto distributions in Figure~\ref{fig:EDF}. The empirical MEP exhibits an approximately linear (increasing) pattern with the threshold, characteristic of heavy-tailed distributions and consistent with the linear MRL behaviour of the Pareto distribution; see \cite{resnick2007heavy}. Figure~\ref{fig:Boxplot} shows a
right-skewed distribution with extreme or high observations above the selected threshold for both datasets. The Pareto fits are obtained using the MLE ($\hat{\alpha}=1.759$) and the MME ($\hat{\alpha}=1.792$). The fitted Pareto distributions based on the MLE and the MME are largely consistent with the empirical distribution.
The $p$-values are reported in Table~\ref{table:pvalue}. They are obtained using a parametric bootstrap method with $B=10{,}000$ replications, where each bootstrap sample is generated from the fitted Pareto distribution with the shape parameter estimated using either the MLE or the MME. The results indicate that none of the tests rejects the null hypothesis, providing no statistical evidence against the assumption that the Danish fire insurance claims exceeding $13.50$ million DKK follow a Pareto distribution.
\subsection{Illustration 2}
\noindent
Air pollution is a growing global concern, with significant adverse effects on public health, including respiratory and cardiovascular diseases, and impacts on overall quality of life \cite{xue2021brain}. Among the major air pollutants, fine particulate matter (PM$_{2.5}$) is of particular concern because of its small particle size, ability to penetrate deep into the respiratory system, and strong association with adverse health effects. PM$_{2.5}$ is also an important component considered in the construction of the Air Quality Index (AQI); motivated by this, we consider daily PM$_{2.5}$ observations in the Delhi region of India. The air quality dataset, available from the GitHub repository \href{https://github.com/Shivshankarnila/MRL-GOF-Test/tree/main}{Delhi Air Quality Dataset}, contains daily AQI measurements along with observations on PM$_{2.5}$, PM$_{10}$, NO$_2$, SO$_2$, CO, and Ozone. Since our interest is in the upper tail of the PM$_{2.5}$ distribution, we consider PM$_{2.5}$ observations for November and December in $2020$ to $2024$, which correspond to the period when the most extreme pollution episodes occur. To focus on extreme pollution events, we select observations exceeding the threshold of $252.52$ and standardize them by defining $Y=X/252.52$, so that $Y\geq1$ for $X\geq252.52$; the resulting exceedance sample, having sample size $26$, is then used to assess the GoF of the Pareto distribution.

The corresponding diagnostic plots are shown in Figure~\ref{fig:EDF}. The fitted Pareto distributions based on the MLE ($\hat{\alpha}_{\mathrm{MLE}}=4.760$) and the MME ($\hat{\alpha}_{\mathrm{MME}}=4.964$) show reasonable agreement with the empirical distribution. The empirical mean excess plot also exhibits an approximately increasing pattern over the lower range of thresholds, broadly consistent with the Pareto model. The $p$-values are reported in Table~\ref{table:pvalue}; none of the tests rejects the null hypothesis, providing no statistical evidence against the Pareto assumption for observations exceeding $252.52$.

\begin{table}[!t]
\centering
\caption{Empirical $p$-values for the real data analysis.}
\label{table:pvalue}
\resizebox{8cm}{!}{
\begin{tabular}{c c c c c c}
\hline
& \multicolumn{2}{c}{Illustration I} &
& \multicolumn{2}{c}{Illustration II} \\
\cline{2-3}\cline{5-6}
Test & MLE & MME & Test & MLE & MME \\
\hline
$\widehat{\Delta}$ & 0.881 & 0.947  & $\widehat{\Delta}$ & 0.336  & 0.350 \\
KS                 & 0.731 & 0.810 & KS                 &  0.308 & 0.282  \\
AD                 &  0.428 & 0.680 & AD                 & 0.061 &  0.073 \\
CvM                &  0.731 & 0.859 & CvM                & 0.338 & 0.322 \\
$G_I$              & 0.772 & 0.774 & $G_I$              &  0.445 & 0.382 \\
$G_M$              & 0.703 &  0.775 & $G_M$              & 0.552 & 0.495 \\
$T_{2,0.5}^{(1)}$  &   0.616 &  0.695 & $T_{2,0.5}^{(1)}$  &  0.465 & 0.415\\
$T_{2,0.5}^{(2)}$  & 0.711 & 0.750 & $T_{2,0.5}^{(2)}$  & 0.441 &  0.396 \\
$ I_{n,2}$          & 0.705 & 0.702 & $ I_{n,2}$          &  0.913 & 0.911 \\
$ZA$             & 0.510 & 0.520 & $ZA$             & 0.854 & 0.847 \\
$ZB$             & 0.363 &  0.383 & $ZB$             & 0.429 & 0.417 \\
\hline
\end{tabular}
}
\end{table}

\begin{figure}[!t]
\centering
\includegraphics[width=1.0\textwidth]{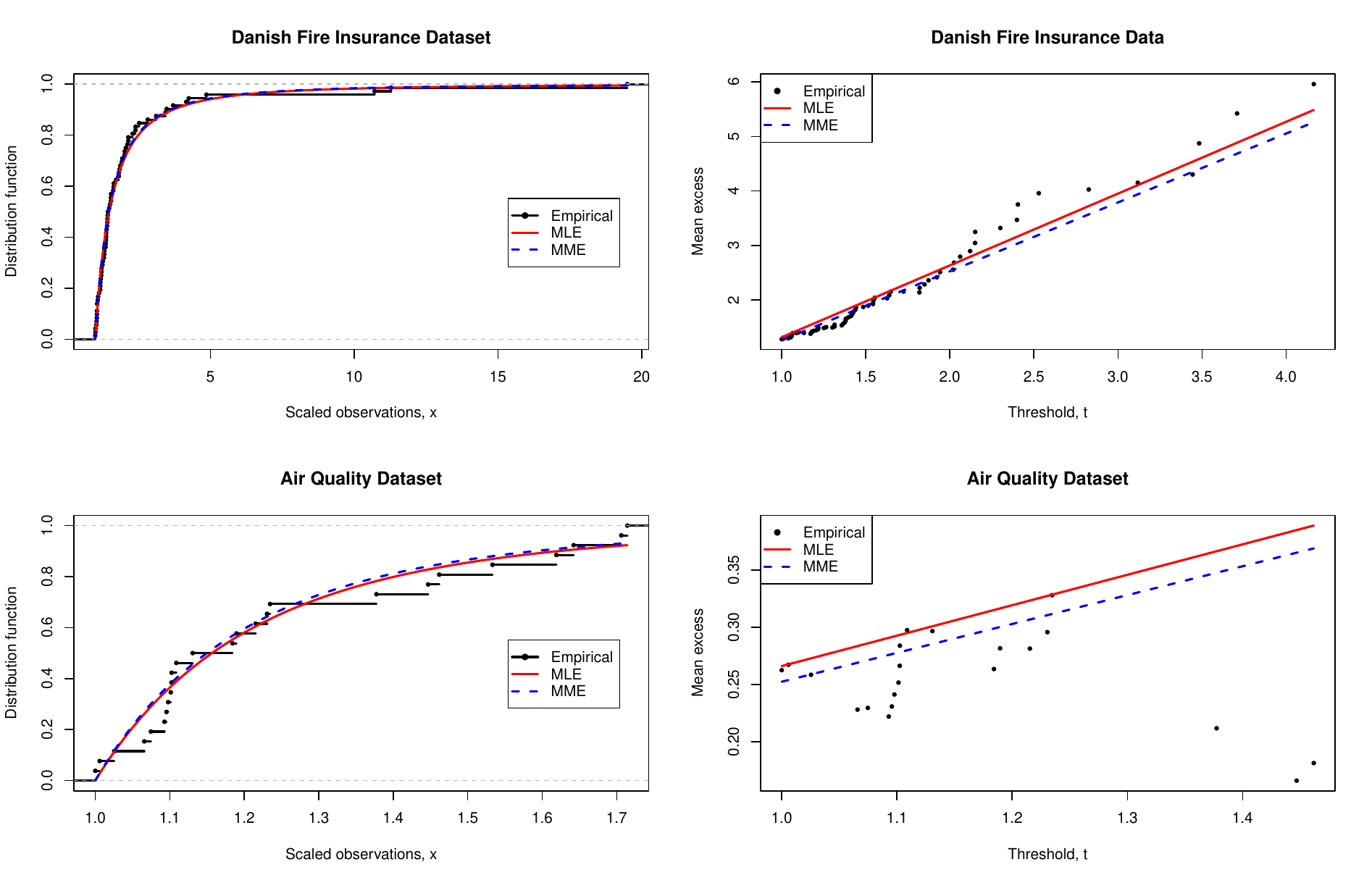}
\caption{Empirical distribution functions and mean excess plots for the scaled above-threshold Danish fire insurance and air quality datasets, together with the fitted Pareto distributions based on the MLE and MME.}
\label{fig:EDF}
\end{figure}
\begin{figure}[!t]
\centering
\includegraphics[width=1.0\textwidth]{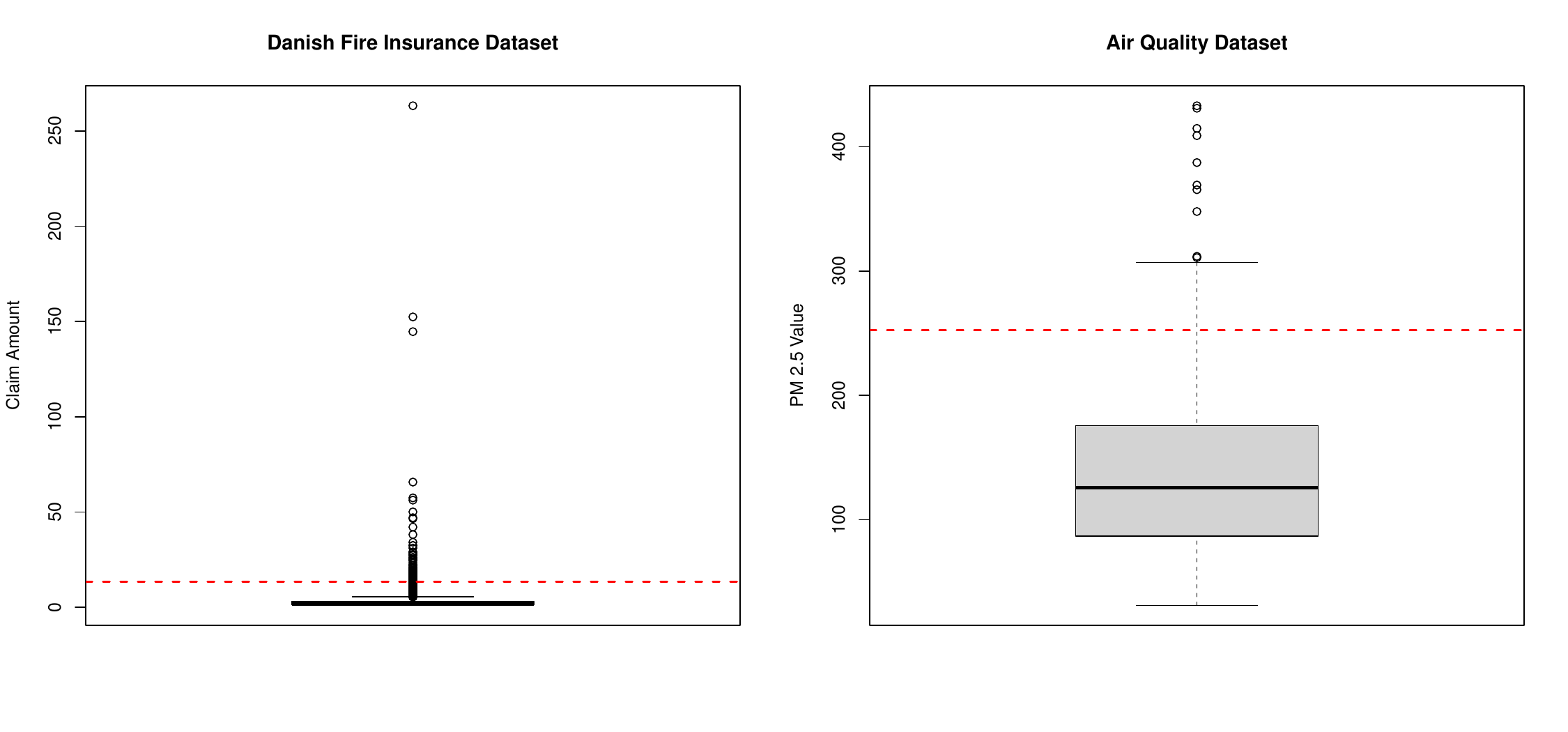}
\caption{Boxplots of the Danish fire insurance and air quality datasets. The dashed horizontal lines indicate the selected thresholds of $13.50$ and $252.52$.}
\label{fig:Boxplot}
\end{figure}
\section{Conclusions and future directions}\label{conclusion}
\noindent In this paper, we proposed a new GoF test for the Pareto distribution based on a characterizations of the MRL function. The asymptotic distribution of the proposed test is obtained using the theory of $U$-statistics.
A comprehensive Monte Carlo simulation study was conducted to evaluate the finite-sample performance of the proposed test across a variety of alternative distributions.
The numerical results show that none of the tests consistently dominates the others or attains the highest power across all alternatives. The findings demonstrate that the proposed test is competitive with existing procedures across the alternatives considered. Finally, we applied the proposed methodology to two important real-world datasets, namely Danish fire claims and Delhi air pollution data, and the results supported the use of the Pareto distribution for the upper tail. 

An interesting direction for future research arises from a fundamental result in Extreme Value Theory (EVT), which states that, under suitable regularity conditions, the distribution of exceedances above a sufficiently high threshold ($u$) can be approximated by a GPD; for more details, see \cite{coles2001introduction,hu2018evmix,scarrott2012review}. A key issue in applying this result is selecting an appropriate threshold for practical EVT applications across a variety of scenarios. This motivates an important GoF problem: given a threshold, or estimated from the data, the question is whether the GPD provides an adequate fit to the exceedances above that threshold. 
In particular, extending the present characterization-based approach to develop a GoF test for the GPD and assess the adequacy of the GPD approximation for threshold exceedances is an interesting direction for future research.
 Another natural direction is the development of GoF procedures for Pareto-type models in which the scale parameter is unknown and must be estimated from the data.
\section*{Declarations}
\noindent\textbf{Funding}\\
The first author acknowledges financial support from the CSIR-UGC (Reference ID Number $191620013874$), Government of India.\\
\noindent \textbf{Conflict of interest} \\
The author declares no conflict of interest.\\
\noindent\textbf{Data availability statement}\\
The datasets used in the current study are available in the GitHub repository: \url{https://github.com/Shivshankarnila/MRL-GOF-Test/tree/main}. \\
 \noindent\textbf{Code availability statement}\\
 The R code is available from the corresponding author upon reasonable request.
 \bibliographystyle{apalike}
 \bibliography{01ref}
\section*{Appendix}
\label{proof:appendix}

\section*{A1. Proof of Theorem \ref{thm1}}
\phantomsection
\makeatletter
\def\@currentlabel{A1}
\makeatother
\label{proof:thm1}
\begin{proof}
To prove the theorem, first consider
\[E\left[\left(X-t-\frac{t}{\alpha-1}\right)I(X>t)\right]=0,\qquad t>1.\]
Then
\[E\left[(X-t)I(X>t)\right]=\frac{t}{\alpha-1}P(X>t).\]
Since
\[E\left[(X-t)I(X>t)\right]
=P(X>t)E(X-t\mid X>t),\]
it follows that
\[E(X-t\mid X>t)=\frac{t}{\alpha-1}.\]
Hence,
\[\begin{aligned}
E(X\mid X>t)
&=t+\frac{t}{\alpha-1}
=\frac{\alpha}{\alpha-1}t
=ct,
\end{aligned}\]
where \(c=\alpha/(\alpha-1)>1\), since \(\alpha>1\) is assumed throughout. By the characterization of the Pareto distribution discussed in
\cite{arnold}, Section~3.9.1, if
\[E(X\mid X>t)=ct,\qquad t>t_0,\]
then
\[X\sim P(I)(t_0,\alpha).\]
Since here \(t_0=1\), we conclude that
\[X\sim P(\alpha).\]

Conversely, suppose that \(X\sim P(\alpha)\), with DF
\[F(x)=1-x^{-\alpha},\qquad x\ge1.\]
The corresponding MRL function is
\[E(X-t\mid X>t)=\frac{t}{\alpha-1},
\qquad t>1.\]
Therefore,
\[E\left[(X-t)I(X>t)\right]
=\frac{t}{\alpha-1}P(X>t),\]
which is equivalent to
\[E\left[\left(X-t-\frac{t}{\alpha-1}\right)I(X>t)\right]=0.\]
This completes the proof.
\end{proof}

\section*{A2. Proof of Theorem \ref{thm2}}
\phantomsection
\makeatletter
\def\@currentlabel{A2}
\makeatother
\label{proof:thm2}
\begin{proof}
To establish the $\widehat{\Delta}$ converges in probability to $\Delta(F)$, we first note that
$U_1$ and $U_2$ are $U$-statistics, they are consistent estimators of
\[\theta_1=E\left(\frac{|X_1-X_2|}{2}\right)
\qquad \text{and} \qquad
\theta_2=E\left(\frac{\min(X_1,X_2)}{2}\right),\]
respectively (see \cite{lehmann1951consistency}). Thus,
\[U_1 \xrightarrow{P} \theta_1
\qquad \text{and} \qquad
U_2 \xrightarrow{P} \theta_2.\]
Recall, as we have,
$\widehat{\alpha}\xrightarrow{P}\alpha$,
 $\alpha>1$, the function $g(x)=1/(x-1)$ is continuous at
$x=\alpha$. Hence, by the continuous mapping theorem (~\citet{van2000asymptotic}, theorem 2.3), we have 
\[\frac{1}{\widehat{\alpha}-1}
\xrightarrow{P}
\frac{1}{\alpha-1}.\]
Therefore, we have
\[\widehat{\Delta}=
U_1-\frac{1}{\widehat{\alpha}-1}U_2
\xrightarrow{P}
\theta_1-\frac{1}{\alpha-1}\theta_2.\]
Since
\[\Delta(F)=\theta_1-\frac{1}{\alpha-1}\theta_2,\]
we obtain
\[\widehat{\Delta}\xrightarrow{P}\Delta(F).\]

\noindent  Hence, $\widehat{\Delta}$ converges in probability to $\Delta(F)$.
\end{proof}

\section*{A3. Proof of Theorem \ref{thm3}}
\phantomsection
\makeatletter
\def\@currentlabel{A3}
\makeatother
\label{proof:thm3}
\begin{proof}
Recall from \eqref{Test_Statistics} that
\[\widehat{\Delta}=U_1-\frac{U_2}{\widehat{\alpha}-1},\]
where $U_1$ and $U_2$ are the $U$-statistics defined in Section~\ref{sec2}, given by
\[\begin{aligned}
U_1 &=
\frac{2}{n(n-1)}
\sum_{i=1}^{n}
\sum_{j=1,j<i}^{n}
\frac{|X_i-X_j|}{2},
\qquad
U_2 &=
\frac{2}{n(n-1)}
\sum_{i=1}^{n}
\sum_{j=1,j<i}^{n}
\frac{\min(X_i,X_j)}{2}.
\end{aligned}\]
Let
$\theta_1=E[h_1(X_1,X_2)]$,
$\theta_2=E[h_2(X_1,X_2)]$, and the corresponding population quantity is
\[\Delta(F)=\theta_1-\frac{\theta_2}{\alpha-1}.\]
We first expand $\widehat{\Delta}-\Delta(F)$, 
\[\begin{aligned}
\widehat{\Delta}-\Delta(F)
&=
U_1-\frac{U_2}{\widehat{\alpha}-1}
-\theta_1+\frac{\theta_2}{\alpha-1}
\\
&=(U_1-\theta_1)
-\frac{U_2}{\widehat{\alpha}-1}
+\frac{\theta_2}{\alpha-1}.
\end{aligned}\]
Now, adding and subtracting $\dfrac{U_2}{\alpha-1}$, we obtain
\[\begin{aligned}
\widehat{\Delta}-\Delta(F)
&=
(U_1-\theta_1)
-\frac{U_2}{\widehat{\alpha}-1}
+\frac{U_2}{\alpha-1}
-\frac{U_2}{\alpha-1}
+\frac{\theta_2}{\alpha-1}
\\
&=
(U_1-\theta_1)+
\left(-\frac{U_2}{\widehat{\alpha}-1}
+\frac{U_2}{\alpha-1}
\right)
+
\left(-\frac{U_2}{\alpha-1}
+\frac{\theta_2}{\alpha-1}
\right)
\\
&=
(U_1-\theta_1)
-\left(\frac{1}{\widehat{\alpha}-1}
-\frac{1}{\alpha-1}
\right)U_2
-\frac{U_2-\theta_2}{\alpha-1}.
\end{aligned}\]
So,
\[
\begin{aligned}
\widehat{\Delta}-\Delta(F)
&=
(U_1-\theta_1)
-\frac{1}{(\alpha-1)}(U_2-\theta_2)
-\left(
\frac{1}{\widehat{\alpha}-1}
-
\frac{1}{\alpha-1}
\right)U_2.
\end{aligned}
\]

Since $\alpha>1$ and $\widehat{\alpha}$ is root-$n$ consistent, a first-order Taylor expansion of
$f(x)=(x-1)^{-1}$ about $\alpha$ gives
\[\frac{1}{\widehat{\alpha}-1}
=\frac{1}{\alpha-1}
- \frac{1}{(\alpha-1)^2}(\widehat{\alpha}-\alpha)
+o_p(n^{-1/2}).\]
Multiplying by $\sqrt{n}$, we obtain
\begin{equation}
\label{eq:h}
\begin{aligned}
\sqrt{n}(\widehat{\Delta}-\Delta(F))
&=
\sqrt{n}(U_1-\theta_1)
-\frac{1}{\alpha-1}\sqrt{n}(U_2-\theta_2)
\\
&\qquad
+\frac{U_2}{(\alpha-1)^2}
\sqrt{n}(\widehat{\alpha}-\alpha)
+o_p(1).
\end{aligned}
\end{equation}

Since $U_2\xrightarrow{P}\theta_2$, by Slutsky's theorem,
\[\begin{aligned}
\sqrt{n}(\widehat{\Delta}-\Delta(F))
&=
\sqrt{n}(U_1-\theta_1)
-\frac{1}{\alpha-1}\sqrt{n}(U_2-\theta_2)
\\
&\qquad
+
\frac{\theta_2}{(\alpha-1)^2}
\sqrt{n}(\widehat{\alpha}-\alpha)
+
o_p(1).
\end{aligned}\]
By the Hoeffding decomposition for U-statistics \cite{lee2019u}, we have
\[\sqrt{n}(U_i-\theta_i)
=\frac{2}{\sqrt n}\sum_{j=1}^{n}
h_i^{(1)}(X_j)
+o_p(1),\qquad i=1,2,\]
where
\[h_i^{(1)}(X_1)
=E\!\left[h_i(X_1,X_2)\mid X_1\right]
-\theta_i,
\qquad i=1,2.\]

Further, under the stated assumptions, the estimator $\widehat{\alpha}$ admits the asymptotic linear representation
\[\sqrt{n}(\widehat{\alpha}-\alpha)
=\frac{1}{\sqrt n}
\sum_{j=1}^{n}
\ell(X_j)
+
o_p(1),\]
where $\ell(X_1)$ is the influence function of $\widehat{\alpha}$. For the MLE,
$\ell(X_1)=\alpha(1-\alpha\log X_1)$.
Hence, substituting the above expressions in \eqref{eq:h}, we obtain
\[\begin{aligned}
\sqrt{n}(\widehat{\Delta}-\Delta(F))
&=
\frac{1}{\sqrt n}
\sum_{j=1}^{n}
\left[
2h_1^{(1)}(X_j)
-\frac{2}{\alpha-1}h_2^{(1)}(X_j)
+\frac{\theta_2}{(\alpha-1)^2}\ell(X_j)
\right]
+o_p(1).
\end{aligned}\]
By the classical central limit theorem, it follows that
\[\sqrt{n}(\widehat{\Delta}-\Delta(F))
\xrightarrow{d}
N(0,\sigma^2),\]
where
\[\sigma^2
=\operatorname{Var}
\left[2h_1^{(1)}(X_1)
-\frac{2}{\alpha-1}h_2^{(1)}(X_1)
+\frac{\theta_2}{(\alpha-1)^2}\ell(X_1)
\right].\]
This completes the proof.
\end{proof}

\section*{A4. Proof of Corollary \ref{corl1}}
\phantomsection
\makeatletter
\def\@currentlabel{A4}
\makeatother
\label{proof:corl1}
\begin{proof}
Under $H_0$, we have $\Delta(F)=0$. From the result established in Theorem~\ref{thm3}, we have
\[\sqrt{n}\,\widehat{\Delta}
=\frac{1}{\sqrt{n}}\sum_{j=1}^{n}
\left[2h_1^{(1)}(X_j)
-\frac{2}{\alpha-1}h_2^{(1)}(X_j)
+\frac{\theta_2}{(\alpha-1)^2}\ell(X_j)
\right]
+o_P(1).\]
Hence, by the classical Central Limit Theorem,
\[\sqrt{n}\,\widehat{\Delta}
\xrightarrow{d}
N(0,\sigma_0^2),\]
where
\[\sigma_0^2
=\operatorname{Var}_{H_0}
\left[2h_1^{(1)}(X_1)
-\frac{2}{\alpha-1}h_2^{(1)}(X_1)
+\frac{\theta_2}{(\alpha-1)^2}\ell(X_1)
\right].\]
Here
\[h_i^{(1)}(X_1)
=E\!\left[h_i(X_1,X_2)\mid X_1\right]-\theta_i,
\qquad i=1,2,\]
with
$\theta_1=E[h_1(X_1,X_2)],~
\theta_2=E[h_2(X_1,X_2)]$,
and $\ell(X_1)$ denotes the influence function of the estimator
$\widehat{\alpha}$. For the MLE,
$\ell(X_1)=\alpha(1-\alpha\log X_1)$.
\noindent This completes the proof.
\end{proof}

\end{document}